%% file: main.tex
\documentclass[journal,draftcls,onecolumn,12pt,twoside]{IEEEtran}
\normalsize

\usepackage{cite}
\usepackage[pdftex]{graphicx}
\usepackage[cmex10]{amsmath}
\usepackage{amssymb}
\usepackage{wasysym}
\usepackage{mathrsfs}
\usepackage{stmaryrd}
\usepackage{upgreek}
\usepackage{enumerate}
\usepackage{array}
\usepackage{color}
\usepackage[tight,footnotesize]{subfigure}
\usepackage{subfig}
\usepackage{algorithm}
\usepackage{algorithmic}

\newcommand{\fig}[4]{ \begin{figure}[#4]
  \centering
   \includegraphics[width=#3\textwidth]{#1}
   \caption{#2}\label{fig:#1}
  \end{figure}
}

\usepackage{soul}

\usepackage{amsthm}

\newtheorem{definition}{Definition}
\newtheorem{prop}{Proposition}

\title{Optimized Rate-Adaptive Protograph-Based LDPC Codes for Source Coding with Side Information}
\author{Fangping Ye$^1$, Elsa Dupraz$^1$, Zeina Mheich$^2$, Karine Amis$^1$\\
\small $^1$ IMT Atlantique, Lab-STICC, UBL, 29238 Brest, France \\ $^2$ University of Surrey, United Kingdom
\thanks{Parts of the materials of Sections II.C and III of this manuscript were published in WCNC 2018~\cite{mheich18WCNC} and ICT 2018~\cite{Ye18ICT}. In introduction, we describe the differences between this work and the results published in~\cite{mheich18WCNC,Ye18ICT}. }}

\begin{document}

\maketitle

\input{Abstract.tex}

\section{Introduction}
\input{Introduction.tex}

\section{Source coding with side information at the decoder}\label{sec:problem}
\input{problem.tex}

\section{LDPC codes for source coding with side information}\label{sec:LDPC_codes}
\input{LDPC_codes.tex}

\input{Method_Binary_case.tex}
 
\section{Simulation Results}\label{sec:simu}
\input{Simulation.tex}

\section{Conclusion}
\input{Conclusion.tex}

\section{Acknowledgement}
This work has received a French government support granted to the Cominlabs excellence laboratory and managed by the National Research Agency in the “Investing for the Future” program under reference ANR-10-LABX-07-01.

\bibliographystyle{IEEEtran}
\bibliography{sample}

\end{document}

%% file: Abstract.tex
\begin{abstract}
This paper considers the problem of source coding with side information at the decoder, also called Slepian-Wolf source coding scheme. In practical applications of this coding scheme, the statistical relation between the source and the side information can vary from one data transmission to another, and there is a need to adapt the coding rate depending on the current statistical relation. In this paper, we propose a novel rate-adaptive code construction based on LDPC codes for the Slepian-Wolf source coding scheme. The proposed code design method allows to optimize the code degree distributions at all the considered rates, while minimizing the amount of short cycles in the parity check matrices at all rates. Simulation results show that the proposed method greatly reduces the source coding rate compared to the standard LDPCA solution.
\end{abstract}

%% file: Introduction.tex
This paper considers the problem of lossless source coding with side information at the decoder, also called Slepian-Wolf source coding~\cite{SW}.
In this scheme, the objective is to reduce the coding rate of the source $X$ by exploiting the side information $Y$ available at the decoder.
This problem has regained attention recently due to its use in many modern multimedia applications like Distributed Source Coding (DSC)~\cite{xiong2004distributed}, Free-Viewpoint Television (FTV)~\cite{FTV,tanimoto2010free}, or Massive Random Access (MRA)~\cite{MRA}. 
For instance, in DSC, the source $X$ represents the data sent by one sensor, and the side information $Y$ is the data already transmitted by other adjacent sensors. 
As another example, FTV is a video system in which the user freely switches from one view to another; in a soccer game, he may decide to follow a player or to focus on the goal. 
In FTV, $X$ is the current requested view, and $Y$ represents the views previously received by the user.
A key aspect of the aforementioned applications is that the source $X$ and the side information $Y$ are correlated. This allows to reduce the source coding rate to $H(X|Y)$ bits/source symbols~\cite{SW} instead of $H(X)$ when no side information is available.

Practical Slepian-Wolf source coding schemes can be constructed from error-correction codes such as Low Density Parity Check (LDPC) codes~\cite{1042242}. LDPC codes were first invented by Gallager~\cite{gallager1962low} in the context of channel coding. 
For very long codewords (more than $10000$ bits), they are known to approach Shannon capacity in channel coding, and to achieve a coding rate close to the conditional entropy in Slepian-Wolf source coding~\cite{chen2009duality}.
For shorter codewords, carefully designed LDPC codes show a reasonable loss in performance compared to Shannon capacity or conditional entropy~\cite{polyanskiy2010channel}.

The performance of an LDPC code depends on its degree distribution, that gives the amount of non-zero values in the code parity check matrix. The code degree distribution can be described either in a polynomial form~\cite{richardson2001design} or by use of a small graph called protograph~\cite{thorpe2003low}.
In this paper, we consider the protograph description as it allows the design of capacity-approaching LDPC codes in a simpler way than polynomial degree distributions~\cite{divsalar2009capacity}. 
At short to medium length, the code performance can also be lowered by short cycles in the code parity check matrix~\cite{han2014construction}.
Therefore, the LDPC code construction is usually realized in two steps. The first step optimizes the protograph for good decoding performance~\cite{richardson2001design,dupraz15Com,542711}. The second step constructs the parity check matrix from the selected protograph by applying a PEG algorithm~\cite{1377521,jiang2014efficient,healy2016design} that lowers the amount of short cycles in the parity check matrix. 

This standard LDPC code construction is adapted to Slepian-Wolf source coding problems with a fixed statistical relation $P(Y|X)$ between the source $X$ and the side information $Y$. 
However, when this statistical relation varies from one data frame to another, such codes may suffer from either rate loss or decoding failure. 
For instance, in FTV, this statistical relation often changes since different sets of views can be available at the decoder, depending on the successive requests of the users. 
This issue can be solved by rate-adaptive LDPC codes that allow to adapt the rate depending on the current statistical relation between $X$ and $Y$.
In channel coding, standard solutions to construct rate-adaptive LDPC codes are puncturing~\cite{ha2004rate} and parity check matrix extension~\cite{yazdani2004construction,van2012design}.
However, in source coding, puncturing leads to a poor decoding performance~\cite{varodayan2006rate} and parity check matrix extension cannot be applied, as it would require to artificially increase the source sequence length. 

In source coding, standard methods to construct rate-adaptive LDPC codes are Rateless codes~\cite{eckford2005rateless,yu2013improved} and Low Density Parity Check Accumulated (LDPCA) codes~\cite{varodayan2006rate}.
Rateless codes start from a low rate code and construct higher rate codes by transmitting a part of the source bits.  
However, the main issue of this method is that it is difficult to construct good low rate LDPC codes~\cite{yang2004design}.
LDPCA takes the opposite approach by starting from a high rate code. Lower rates are then obtained by puncturing accumulated syndrome bits rather than puncturing the syndrome bits directly.
Unfortunately, the LDPCA accumulator has a fixed regular structure, and it cannot be optimized in order to obtain good degree distributions or to avoid short cycles that could degrade the code performance at lower rates.   
In order to improve the LDPCA construction, it was proposed in~\cite{cen2009design} to consider a non-regular accumulator optimized for any rate of interest. 
The method of~\cite{cen2009design} optimizes the non-regular accumulator for codes with large length, but it does propose any finite-length code construction that would allow to reduce the amount of short cycles in the low-rate codes.
As an intermediate solution,~\cite{kasai2010rate} starts with an initial rate $1/2$ and applies Rateless codes construction to increase the rate and LDPCA codes construction to decrease the rate.  
This permits to avoid the main issue of Rateless codes, but the shortage of LDPCA codes still exists.

In this paper, we consider the intermediate solution of~\cite{kasai2010rate} and we replace the LDPCA part by an alternative rate-adaptive LDPC code construction that was initially introduced in~\cite{Ye18ICT,mheich18WCNC}. 
This construction replaces the LDPCA accumulator by intermediate graphs that combine the syndrome bits in order to obtain lower rate codes. 
In this construction, the intermediate graphs must be full-rank in order to allow the code to be rate-adaptive.

The code design method initially proposed in~\cite{mheich18WCNC,Ye18ICT} only consider unstructured finite-length code constructions, that is without design of the degree distributions of the lower rate codes.
Here, we introduce a novel design method that allows to select the photographs of the intermediate graphs so as to optimize the decoding performance of all the codes constructed at all rates of interest.
We also propose a new algorithm called Proto-Circle that constructs the intermediate graphs according to their protographs, while minimizing the amount of short cycles in the codes at all the considered rates. 
In addition, the rate-adaptive construction of~\cite{mheich18WCNC,Ye18ICT} only permits to consider a small number of rates, in the order of magnitude of the size of the protograph. 
The code design method we propose in this paper permits to obtain a much larger number of rates, in the order of magnitude of the size of the parity check matrix. 
Our simulation results show that the proposed rate-adaptive LDPC code construction provides improved performance compared to LDPCA. This comes from the careful protograph selection and from the fact that our method reduces the number of short cycles in the constructed codes at all rates.

The outline of this paper is as follows. Section~\ref{sec:problem} presents the Slepian-Wolf source coding problem. Section~\ref{sec:LDPC_codes} describes existing rate-adaptive LDPC code construction for source coding. Section~\ref{sec:ex_ra_const} restates the rate-adaptive construction of~\cite{mheich18WCNC,Ye18ICT}. Section~\ref{sec:int_matrix} introduces our code design method by describing the protograph optimization and the Proto-Circle algorithm. Section~\ref{sec:several_rates} presents our solution to consider an increased number of rates. To finish, Section~\ref{sec:simu} shows our simulation results. 

%% file: problem.tex
\begin{figure*}[t]
\begin{center}
   \subfigure[~]{  \includegraphics[width=.4\linewidth]{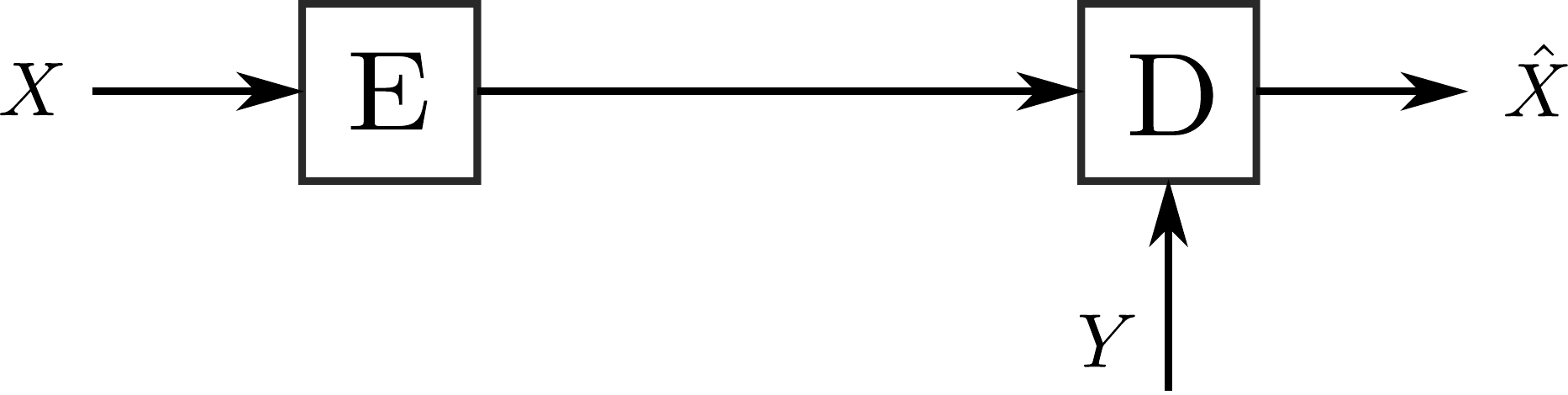}}
   \subfigure[~]{  \includegraphics[width=.4\linewidth]{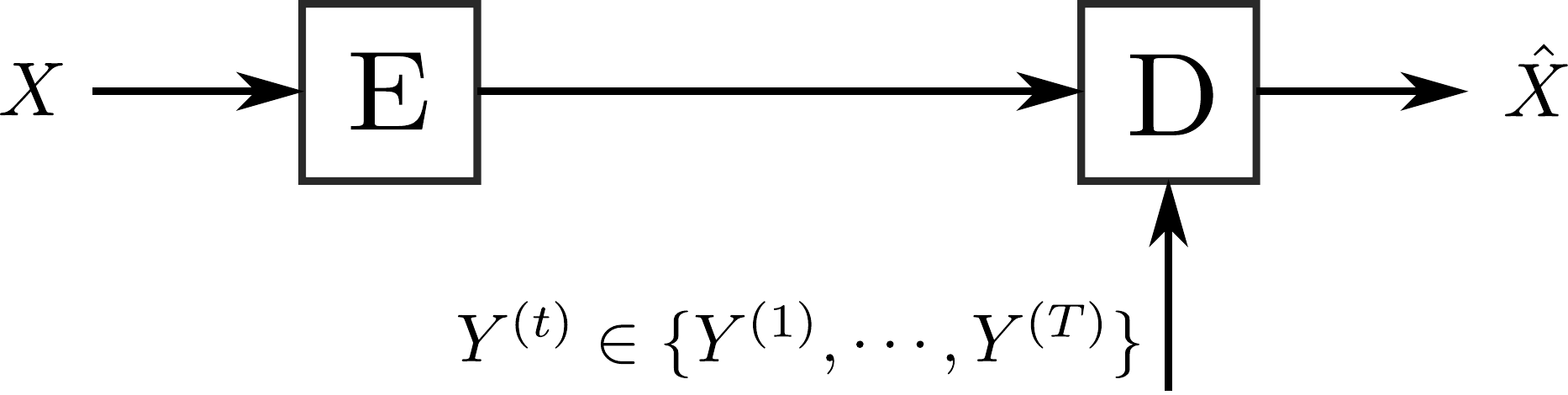}}
\end{center}
\caption{(a) Slepian-Wolf source coding, (b) Source coding with several possible side informations at the decoder }
\label{fig:sc_si}
\end{figure*}

This section describes the lossless coding of a source $X$ with side information $Y$ available at the decoder. 
The source $X$ generates independent and identically distributed (i.i.d.) symbols $X_1,\cdots, X_n$. 
For all $i \in \{1,\cdots, n\}$, we assume that the source symbol $X_i$ takes values in a binary alphabet $\{0,1\}$. 
The probability mass function of the source $X$ is denoted by $P(X)$.
The source bits $X_i$ may represent the pixels of an image or quantized sensor measurements. 

In the original SW source coding scheme~\cite{SW} depicted in Figure~\ref{fig:sc_si} (a), a side information $Y$ is available at the decoder and helps the reconstruction of $X$. The side information $Y$ generates i.i.d. symbols $Y_1,\cdots, Y_n$ and the symbols $Y_i$ belong to an alphabet $\mathcal{Y}$ that can be either discrete or continuous. 
With a slight abuse of notation, we denote by $P(Y|X)$ the correlation channel between $X$ and $Y$. If $\mathcal{Y}$ is discrete, $P(Y|X)$ is a conditional probability mass function. If  $Y$ is continuous, $P(Y|X)$ is a conditional density.
For instance, if the correlation channel between $X$ and $Y$ is a Binary Symmetric Channel (BSC), then $\mathcal{Y} = \{0,1\}$, $P(Y=1|X=0) = P(Y=0|X=1) = p$, and $p$ is the crossover probability of the BSC.
According to~\cite{SW}, the minimum achievable rate for lossless SW source coding is given by $R = H(X|Y)$ bits/source symbol.
The side information $Y$  reduces the coding rate, since $H(X|Y) \leq H(X)$.

Alternatively, in this paper, we consider the problem of source coding with several possible side informations at the decoder~\cite{sgarro1977source}, see Figure~\ref{fig:sc_si} (b).
In this setup, the decoder has access to one side information $Y^{(t)}$ that belongs to a set of $T$ possible side informations $\{Y^{(1)}, \cdots Y^{(T)}\}$. 
Each source $Y^{(t)} \in \{Y^{(1)}, \cdots Y^{(T)}\}$ generates a sequence of $n$ i.i.d symbols $Y_1^{(t)},\cdots, Y_n^{(t)}$, but the decoder only observes one of these length-$n$ sequences. 
Whatever $t \in \{1,\cdots, T\}$, the side information symbols $Y_i^{(t)}$ belong to the same alphabet $\mathcal{Y}$. However, each of the side informations $Y^{(t)}$ corresponds to a different correlation channel described by the conditional probability distribution $P(Y^{(t)}|X)$. 
For instance, each correlation channel $P(Y^{(t)}|X)$ may correspond to a BSC with a different crossover probability $p_t$.

For the setup of Figure~\ref{fig:sc_si} (b), the minimum achievable rate can be evaluated in two different ways. 
If the encoder does not know the index $t$ of the side information available at the decoder, then the minimum achievable rate is given by $R = \displaystyle \max_{t=1,\cdots, T} H(X|Y^{(t)})$~\cite{sgarro1977source,Draper07}.
Otherwise, if the encoder has access to index $t$ by means of e.g. a feedback channel, the minimum achievable rate depends on $t$ and it is given by $R_t = H(X|Y^{(t)})$~\cite{Draper04,yang10IT}. 

In addition,~\cite{MRA} describes the rates $R_t$ as transmission rates from the server (encoder part) to the users (decoder part) and proposes to also take into account the storage rate $Q$ on the server. 
Indeed, in order to achieve transmission rates $R_t$, one could consider storing one different codeword per possible side information $Y^{(t)}$, which would give $Q = \sum_{t=1}^T H(X|Y^{(t)})$ bits/source symbol.
However,~\cite{MRA} proposes an information-theoretic code construction that allows to construct one single incremental codeword at storage rate $Q =  \displaystyle \max_{t=1,\cdots, T} H(X|Y^{(t)}) \leq \sum_{t=1}^T H(X|Y^{(t)})$ from which we can extract $T$ subcodewords with rates $R_t = H(X|Y^{(t)})$ depending on the side information $Y^{(t)}$ available at the decoder.
In this paper, we propose a practical rate-adaptive coding scheme based on LDPC codes which provides such incremental codeword construction.

%% file: LDPC_codes.tex
For the two source coding problems described in Section~\ref{sec:problem}, LDPC codes provide efficient practical coding schemes that perform close to the conditional entropies $H(X|Y)$ or $H(X|Y^{(t)})$ when the codeword length $n$ is large~\cite{chen2009duality}. 
In this section, we first consider the case with one possible side information $Y$ available at the decoder, and we describe the standard construction of LDPC codes with good performance for the Slepian-Wolf source coding scheme. 
We then review existing rate-adaptive LDPC-based solutions that provide practical source coding schemes when side information $Y^{(t)}$ available at the decoder comes from a set  $\{Y^{(1)}, \cdots Y^{(T)}\}$.

\subsection{LDPC codes for Slepian-Wolf source coding}
In this section, we consider the Slepian-Wolf source coding setup in which one single side information $Y$ is available at the decoder. 
%
Let $\underline{x}^{n}=(x_{1},x_{2},\cdots x_{n})^{T}$ stand for a source vector of length $n$ to be transmitted to the decoder. Consider an LDPC parity check matrix $H$ of size $m\times n$ $(m<n)$ and coding rate $R=m/n$. The matrix $H$ is sparse and its non-zero components are all equal to $1$. 
The codeword or syndrome $\underline{c}^{m}=\left(c_{1},c_{2},\cdots c_{m}\right)^{T}$ that is transmitted to the decoder is calculated from $\underline{x}^n$ and $H$ as~\cite{1042242}
  \begin{equation}\label{eq:symdrom_comp}
  \underline{c}^{m}=H \underline{x}^{n} .
  \end{equation}
  In the binary matrix multiplication of~\eqref{eq:symdrom_comp}, additions correspond to XOR operations and multiplications correspond to AND operations.
Once it receives syndrom $\underline{c}^{m}$, the decoder produces an estimate $\hat{\underline{x}}^n$ of $\underline{x}^{n}$ by applying the Belief Propagation algorithm (BP) to $\underline{c}^{m}$ and the side information vector $\underline{y}^n$~\cite{1042242,dupraz15Com}.

The parity check matrix $H$ can be represented by a Tanner Graph that connects $n$ Variable Nodes (VN) $X_1,\cdots,X_n$ with $m$ Check Nodes (CN) $C_1,\cdots,C_m$. There is an edge between a VN $X_i$ and a CN $C_j$ if there is a non-zero value at the corresponding matrix position $H_{i,j}$. 
The decoding performance of LDPC codes highly depends on the choice of the parity check matrix $H$, as we now describe.

\subsection{LDPC code construction} \label{subsec:code_construction}
A parity check matrix $H$ can be constructed from a code degree distribution described by a protograph~\cite{thorpe2003low}. 
A protograph $\mathcal{S}$ is a small Tanner Graph of size $S_{m}\times S_{n}$ with $S_m/S_n=m/n = R$. Each row (respectively column) of $\mathcal{S}$ represents a type of CN (respectively of VN). The protograph $\mathcal{S}$ thus describes the number of connections between $S_{n}$ different types of VNs and $S_{m}$ different types of CNs. A parity check matrix $H$ can be generated from a protograph $\mathcal{S}$ by repeating the protograph structure $Z$ times such that $n = Z S_n $, and by interleaving the connections between the VNs and the CNs of the corresponding types. The interleaving is realized so as to obtain a connected Tanner graph that satisfies the number of connections defined by the protograph. 
It can be done by a PEG algorithm~\cite{1377521} that not only permits to satisfy the protograph constraints, but also to lower the number of short cycles that could severely degrade the decoding performance of the matrix $H$.

\fig{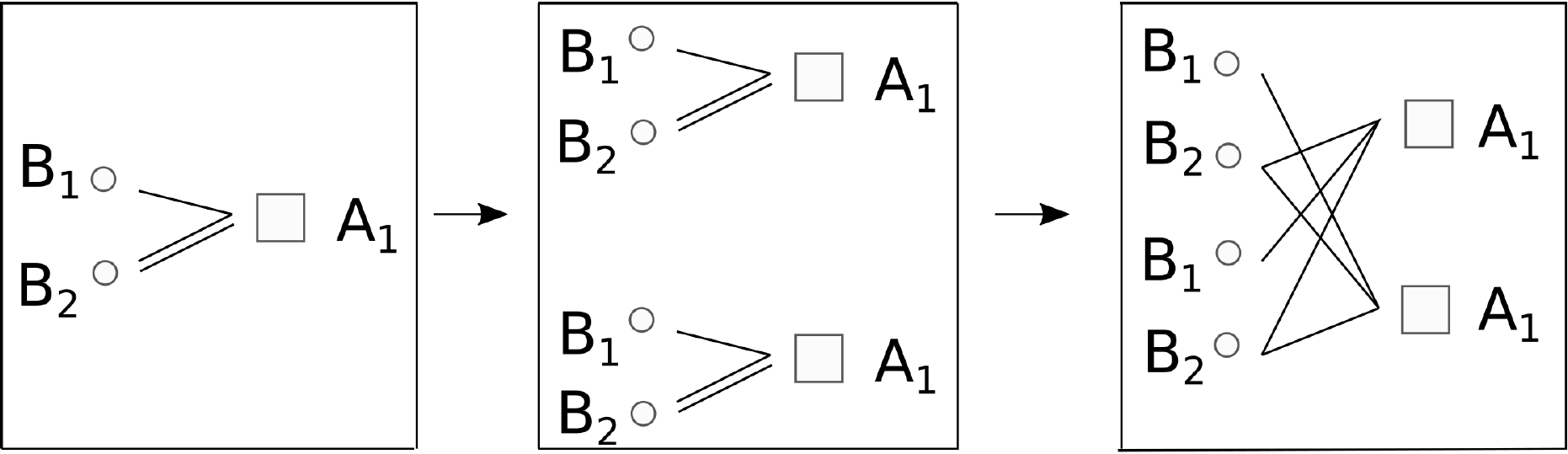}{Construction of a parity check matrix $H$ of size  $2 \times4$ (right picture) from a protograph $\mathcal{S}$ of size $1 \times2$ (left picture)}{0.75}{t}

We now give an example of construction of a parity check matrix $H$ from the protograph
\begin{equation}
\mathcal{S} = \begin{bmatrix} 1 & 2 \end{bmatrix} 
\end{equation}
that represents the connections between one CN of type $A_1$ and two VNs of types $B_1$ and $B_2$. 
The Tanner graph of this protograph is represented in Figure~\ref{fig:protograph} (left part). In order to construct a parity check matrix $H$, the protograph is first duplicated $Z=2$ times (middle part of Figure~\ref{fig:protograph}), and the edges are then interleaved (right part of Figure~\ref{fig:protograph}).
In the final Tanner graph, one can verify that each VN of type $B_1$ is connected to one CN of type $A_1$, and each VN of type $B_2$ is connected to two CNs of type $A_1$.

The performance of a given parity check matrix $H$ highly depends on its underlying protograph $\mathcal{S}$. 
Density Evolution~\cite{910577} permits to evaluate the theoretical threshold of a protograph under asymptotic conditions. The threshold is the worst correlation channel parameter that allows for a decoding error probability $P_e=0$, given that the codeword length tends to infinity. It can thus be used as an optimization criterion to select the protograph.
For a given coding rate $R = m/n$, Differential Evolution~\cite{542711} is an optimization method that permits to find a protograph with a very good threshold.  

In order to optimize the protograph, we fix the size $S_m \times S_n$ of the protograph and impose a maximum degree value $d_{\max}$ for the protograph coefficients.
Differential Evolution is a genetic algorithm that starts with an initial population of $V$ protographs and recombines the elements of the population in order to get new protographs. 
The protographs with the best thresholds (evaluated with density evolution) are then retained in order to form a new population. 
This process is repeated over $L$ iterations. 
The recombination operation proposed in~\cite{542711} stands for real vectors, while the protograph coefficients are discrete. 
In our optimization, we keep the recombination operation of~\cite{542711} and simply round each obtained real value to the closest integer.

\subsection{Rate-adaptive LDPC codes}
In the above LDPC code construction, the coding rate $R$ is fixed once for all. 
But if the available side information $Y^{(t)}$ comes from a set of possible side informations $\{Y^{(1)},\cdots, Y^{(T)}\}$, sending the data at fixed rate $R$ will cause either a rate loss or a decoding failure. 
This is why we now describe rate-adaptive LDPC code constructions that allow to adapt the coding rate depending on the side information $Y^{(t)} $ available at the decoder.

The rate-adaptive Rateless scheme~\cite{eckford2005rateless,yu2013improved} starts by constructing an initial low-rate LDPC code. If a higher rate is needed, a part of the source bits $\underline{x}^{n}$ will be sent in addition to the syndrome $\underline{c}^{m}$. However, the major drawback of the rateless scheme is that it difficult to construct good low rate LDPC codes~\cite{yang2004design,shin2012design}\footnote{Low rate LDPC codes for source coding correspond to high rate LDPC codes for channel coding}, and a bad initial low rate code will cause poor performance at any considered higher rate. Therefore, it is not desirable to apply the Rateless construction from very low rates.

On the opposite, the LDPCA scheme~\cite{varodayan2006rate} starts from a high-rate LDPC code.
It then computes new accumulated symbols $\underline{a}^m = [a_{1},a_{2},\cdots,a_{m}]^T$ from the syndrome $\underline{c}^m$~\eqref{eq:symdrom_comp}  as
\begin{align}\notag
a_{1} & =c_{1},\\ 
a_{i} &=a_{i-1}+c_{i}, ~~ \forall i=\left\{ 2,\cdots,m\right\} , \label{eq:acc_structure}
\end{align}
where the binary sum in~\eqref{eq:acc_structure} correspond to XOR operations.
 If a lower rate is demanded, only a part of the symbols $\left(a_{1},a_{2,}\cdots,a_{m}\right)$ will be sent. For instance, if the original rate is $R$ and a rate $R/2$ is demanded, only the even symbols $a_2, a_4, a_6, \cdots$ will be transmitted.  
 The decoder will then compute all the differences $a_{i} - a_{i-2} = c_{i} + c_{i-1}$, before applying a BP decoder in order to estimate the source vector $\mathbf{x}^n$ from all the obtained XOR sums $c_{i} + c_{i-1}$. 
 In this construction, puncturing the source symbols $a_i$ rather than the syndrome bits $c_i$ was shown to better preserve the code structure and to greatly improve the decoding performance~\cite{varodayan2006rate}.  
 However, in the LDPCA construction, the accumulator structure~\eqref{eq:acc_structure} is fixed and does not allow for an optimization of the combinations of syndrome symbols $c_i$ that are used by the decoder. The accumulator structure may in particular induce short cycles in the lowest rates and eliminate some source bits from the CN constraints. 
In~\cite{cen2009design}, the LDPCA structure is improved by considering a non-regular accumulator. The non-regular accumulator is designed for any rate of interest by optimizing its polynomial degree distribution under asymptotic conditions. Unfortunately,~\cite{cen2009design} does not propose any finite-length code construction that could solve the short cycles and VN elimination issues.
 
Due to the drawbacks of Rateless and LDPCA schemes, an intermediate solution was proposed in \cite{kasai2010rate}. It first constructs an initial code of rate $R=1/2$. It then applies either the LDPCA method to obtain rates lower than $1/2$ or the Rateless method for rates higher than $1/2$. In this way, the shortage of the Rateless construction can be avoided, but the drawbacks of LDPCA remain. 
In this paper, we thus propose a novel rate-adaptive construction that replaces the LDPCA part in the solution of~\cite{kasai2010rate}.
Our rate-adaptive code design method is based both on an asymptotic performance analysis and on a finite length code construction that permits to avoid short cycles. It is thus well adapted to the construction of short length LDPC codes. 

%% file: Method_Binary_case.tex
\section{Rate-adaptive code construction}\label{sec:ex_ra_const}
The rate-adaptive code design method we propose in this paper is based on a rate-adaptive code structure initially proposed in~\cite{Ye18ICT,mheich18WCNC}.
For the sake of clarity, this section describes the rate-adaptive code structure of~\cite{Ye18ICT,mheich18WCNC}.
This construction starts from a mother code of the highest rate and then builds a sequence of daughter codes of lower rates.
This section only describes the construction of one code of rate $R_2$ from a code of rate $R_1 > R_2$. 
This construction is generalized to more rates later in the paper.

\subsection{Rate-adaptive code construction}\label{sec:mother_daughter}
In the rate-adaptive construction of~\cite{Ye18ICT,mheich18WCNC}, the mother code is described by a parity check matrix $H_1$ of size $m_1\times n$ with coding rate $R_1=m_1/n$.
The Tanner graph $\mathcal{T}_1$ connects the $n$ VNs $\mathcal{X}=\{X_1,\cdots, X_n\}$ to $m_1$ CNs $\mathcal{C} = \{C_1,\cdots, C_{m_1}\}$.
The matrix $H_1$ is constructed from a protograph $\mathcal{S}_1$ according to the code design method described in Section~\ref{sec:LDPC_codes}. 
From the mother matrix $H_1$, we want to construct a daughter matrix $H_2$ of size $m_2\times n$, with $m_2 < m_1$, and rate $R_2 = m_2/n < R_1$.
The Tanner graph $\mathcal{T}_2$ will connect the $n$ VNs $\mathcal{X}$ to $m_2$ CNs $\mathcal{U} = \{U_1,\cdots,U_{m_2} \}$.

\fig{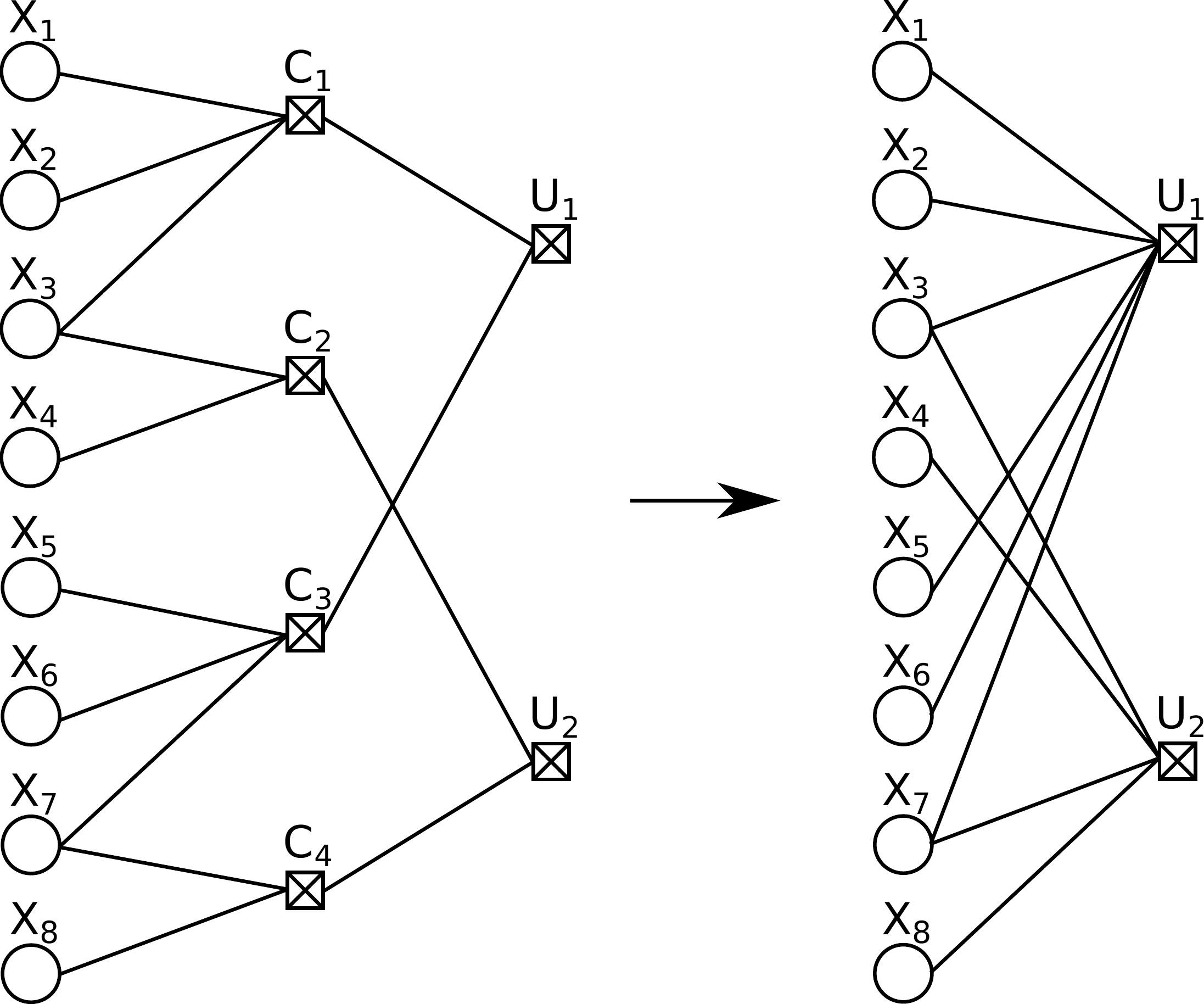}{The left part of the figure shows the combination of $\mathcal{T}_1$  with $\mathcal{T}_{1 \rightarrow 2}$. The right part of the figure shows the resulting $\mathcal{T}_2$. Here, the matrix $H_{1\rightarrow 2}$ is full rank, and one may choose between $\mathcal{C'} = \{ c_1,c_2 \}$, $\mathcal{C'} = \{ c_3,c_4 \}$, $\mathcal{C'} = \{ c_1,c_4 \}$, or $\mathcal{C'} = \{ c_2,c_3 \}$. }{0.4}{t}

In the considered construction, the daughter matrix $H_2$ and the mother matrix $H_1$ are linked by an intermediate matrix  $H_{1\rightarrow 2}$ of size $m_2 \times m_1$ such that
 \begin{equation}\label{eq:defH2}
  H_2 = H_{1\rightarrow 2} H_1 .
 \end{equation}
The Tanner graph $\mathcal{T}_{1 \rightarrow 2}$ of $H_{1\rightarrow 2}$ connects the $m_1$ CNs $\mathcal{C}$ of $\mathcal{T}_1$ to the $m_2$ CNs $\mathcal{U}$ of $\mathcal{T}_2$.
Figure~\ref{fig:graphes} shows an example of the construction of $\mathcal{T}_2$ from $\mathcal{T}_1 $ and $\mathcal{T}_{1 \rightarrow 2}$.
 Note that LDPCA codes can be seen as a particular case of this construction. 
 The intermediate matrix $H_{1\rightarrow 2}$ should be chosen not only to give a good decoding performance for $H_2$, but also to allow $H_1$ and $H_2$ to be rate-adaptive in a sense we now describe.

\subsection{Rate adaptive condition}\label{sec:ra_cond}
 In the construction of~\cite{Ye18ICT,mheich18WCNC}, the following transmission rules are set in order to allow $H_1$ and $H_2$ related by~\eqref{eq:defH2} to be rate-adaptive. 
In order to get a rate $R_2$, we simply transmit all the syndrome values $\underline{u}^{m_2}$, which corresponds to $m_2$ equations defined by the set $\mathcal{U}$. The decoding is then realized with the matrix $H_2$. 
In order to get a rate $R_1$, we transmit all syndrome values in $\underline{u}^{m_2}$ but also a subset $\mathcal{C}' \subseteq \mathcal{C}$ of size $m_1 - m_2$ of the values in $\underline{c}^{m_1}$. This guarantees that the code construction is incremental and that the storage rate is given by $R_1 =\max(R_1,R_2) < R_1 + R_2$.
However, in order to use the matrix $H_1$ for decoding, the receiver must be able to recover the full syndrome $\underline{c}^{m_1}$ from $\underline{u}^{m_2}$ and $\mathcal{C'}$. 
The code that results from the choice of ($H_1$, $H_{1\rightarrow 2}$, $\mathcal{C'}$) is thus said to be rate-adaptive if is satisfies the following condition.

 \begin{definition}[\cite{Ye18ICT,mheich18WCNC}]\label{def:rc_cond}
  The sets $\mathcal{U}$ and $\mathcal{C'}$ define a system of $m_1$ equations with $m_1$ unknown variables $\mathcal{C}$. If this system has a unique solution, then the triplet ($H_1$, $H_{1\rightarrow 2}$, $\mathcal{C'}$) is said to be a rate-adaptive code. 
 \end{definition}
 
The following proposition gives a simple condition that permits to verify whether a given intermediate matrix $H_{1\rightarrow 2}$ gives a rate-adaptive code.

\begin{prop}[\cite{Ye18ICT,mheich18WCNC}]\label{prop:fullrank}
  If the matrix $H_{1\rightarrow 2}$ is full rank, then there exists a set $\mathcal{C'} \subseteq \mathcal{C}$ of size $m_1 - m_2$ such that ($H_1$, $H_{1\rightarrow 2}$, $\mathcal{C'}$) is a rate-adaptive code. 
 \end{prop}
 
 The above proposition shows that if $H_{1\rightarrow 2}$ is full rank, it is always possible to find a set $\mathcal{C'}$ that ensures that $H_1$ and $H_2$ are rate-adaptive. 
The decoding performance of $H_1$ does not depend on the choice of the set $\mathcal{C'}$, since at rate $R_1$, the decoder uses $H_1$ and at rate $R_2$, the decoder uses $H_2$. 
On the opposite, according to~\eqref{eq:defH2}, the decoding performance of the matrix $H_2$ heavily depends on the matrix $H_{1\rightarrow 2}$. 
In~\cite{mheich18WCNC}, the matrix $ H_{1\rightarrow 2}$ is constructed from an exhaustive search, which is hardly feasible when the codeword length increases (from $100$ bits). 
In~\cite{Ye18ICT}, a more efficient method is proposed to construct the intermediate matrix $H_{1\rightarrow 2}$ so as to avoid short cycles in $H_2$. However, the method of~\cite{Ye18ICT} does not optimize the theoretical threshold of the degree distribution of $H_2$, which also influences the code performance.   
In this paper, we propose a novel method based on protographs for the design of the intermediate matrix $H_{1\rightarrow 2}$. This novel method not only allows to optimize the threshold of the protograph of $H_2$, but also to reduce the amount of short cycles in $H_2$. 

\section{Intermediate matrix construction}\label{sec:int_matrix}
This section describes our novel method for the construction of the intermediate matrix $H_{1\rightarrow 2}$ introduced in Section~\ref{sec:ex_ra_const}.
The proposed construction seeks to minimize the protograph threshold at rate $R_2$, and also to reduce the amount of short cycles in the parity check matrix $H_2$.

\subsection{Protograph $\mathcal{S}_2$ of parity check matrix $H_2$}\label{sec:proto_H2}

In order to construct a good parity check matrix $H_{2}$ from the initial matrix $H_{1}$, we first want to select a protograph $\mathcal{S}_2$ with a good theoretical threshold. 
In this section, we consider the following notation. 
Generally speaking, consider the protograph $\mathcal{S}_g$ of size $S_{m_g} \times S_{n_g}$ associated with the matrix $H_g$, where $g \in \{1,2,1\rightarrow 2\}$. 
As a particular case, note that $S_{m_{1\rightarrow 2}} =  S_{m_2}$ and $S_{n_{1\rightarrow 2}} =  S_{m_1}$.
For all $ (i,j) \in \{ 1,\cdots, S_{m_g} \} \times \{ 1,\cdots, S_{n_g} \}$, denote by $s_{i,j}^{(g)}$ the coefficient at the $i$-th row, $j$-th column of $\mathcal{S}_{g}$.
In the protograph $\mathcal{S}_g$, the CN types are denoted $A_1^{(g)}, \cdots,  A_{S_{m_g}}^{(g)}$ and the VN types are denoted $B_1^{(g)}, \cdots,  B_{S_{n_g}}^{(g)}$. 
In the parity check matrix $H_g$, the set of CNs of type $A_i^{(g)}$ is denoted $\mathcal{A}_i^{(g)}$ and the set of VNs of type $B_j^{(g)}$ is denoted $\mathcal{B}_j^{(g)}$. 
Finally, denote by $\underline{h}_k^{(g)}$ the $k$-th row of $H_g$, and denote by  $h_{\ell,k}^{(g)}$ the coefficient at the $\ell$-th row, $k$-th column of $H_{g}$.

Based on the above notation, the following proposition gives the relation between the three protographs $\mathcal{S}_1$, $\mathcal{S}_2$, and $\mathcal{S}_{1\rightarrow2}$.

\begin{prop}\label{prop:rel_proto}
Consider a matrix $H_1$ with protograph $\mathcal{S}_1$ of size $S_{m_1} \times S_n$, a matrix $H_{1\rightarrow 2}$ with protograph  $\mathcal{S}_{1\rightarrow2}$ of size $S_{m_2} \times S_{m_1}$, and a matrix $H_2 = H_{1\rightarrow 2} H_1$. 
Also consider the following two assumptions:
\begin{enumerate}
 \item Type structure: for all $j \in \{1,\cdots, S_{m_1} \}$, $\mathcal{B}_j^{(1\rightarrow 2)} = \mathcal{A}_j^{(1)}$.
 \item No VN elimination: For all $\ell \in \{1,\cdots, m_2\}$, denote by $\mathcal{N}_{\ell}^{(1\rightarrow 2)}$ the positions of the non-zero components in $\underline{h}_{\ell}^{(1\rightarrow 2)}$. Then, $\forall k_1,k_2 \in \mathcal{N}_{\ell}^{(1\rightarrow 2)}$ such that $k_1 \neq k_2$, and $\forall i \in \{1,\cdots,n\}$, $h_{k_1,i}^{(1)} \neq h_{k_2,i}^{(1)}$. 
\end{enumerate}
If these two assumptions are fulfilled, then the matrix
\begin{equation}\label{eq:relation_proto}
\mathcal{S}_{2}=\mathcal{S}_{1\rightarrow2} \mathcal{S}_{1}
\end{equation}
is of size $S_{m_2} \times S_n $ and it is a protograph of the matrix $H_2$.
The operation in~\eqref{eq:relation_proto} corresponds to standard matrix multiplication over the field of real numbers. 
\end{prop}
\begin{proof}
In this proof, for clarity, we denote by $\bigoplus$ the modulo two sums and by $\sum$ the standard sums over the field of real numbers.
With the above notation, relation~\eqref{eq:defH2} can be restated row-wise as
\begin{equation}\label{eq:Intermidiateh} 
 \underline{h}_{\ell}^{(2)}  = \bigoplus_{k=1}^{m_1} h_{\ell,k}^{(1\rightarrow 2)} \underline{h}_{k}^{(1)} 
  = \bigoplus_{j=1}^{S_{m_1}} \underset{\text{ s.t. } h_{\ell,k}^{(1\rightarrow 2)}\neq 0}{\bigoplus_{k \in \mathcal{B}_j^{(1\rightarrow 2)}}} \underline{h}_{k}^{(1)} .
\end{equation}
Relation~\eqref{eq:Intermidiateh} depends on index $\ell$ only through $h_{\ell,k}^{(1\rightarrow 2)}$. This implies that, in~\eqref{eq:Intermidiateh} the type combination is the same for every $\ell \in \mathcal{A}_i^{(1\rightarrow 2)}$. As a result, for all $i \in \{1,\cdots, S_{m_2}\}$, $\mathcal{A}_i^{(2)} = \mathcal{A}_i^{(1\rightarrow 2)} $.
In the same way, deriving relation~\eqref{eq:defH2} column-wise permits to show that $\forall j \in \{1,\cdots, S_n\}$, $\mathcal{B}_j^{(2)} =  \mathcal{B}_j^{(1)}$. 

Now consider $i \in \{1,\cdots, S_{m_2} \}$, $v \in \{1,\cdots, S_{n} \}$, and $\ell \in \mathcal{A}_i^{(2)}$. Then, from~\eqref{eq:Intermidiateh},
\begin{equation}
  s_{i,v}^{(2)}  = \sum_{u \in \mathcal{B}_v^{(2)}} h_{l,u}^{(2)} =  \sum_{u \in \mathcal{B}_v^{(2)}} \left( \bigoplus_{j=1}^{S_{m_1}} \underset{\text{ s.t. } h_{\ell,k}^{(1\rightarrow 2)}\neq 0}{\bigoplus_{k \in \mathcal{B}_j^{(1\rightarrow 2)}}} h_{k,u}^{(1)}  \right) .
\end{equation}
In the vector $\underline{h}_{k}^{(1)}$ with $k \in \mathcal{B}_j^{(1\rightarrow 2)}$, there are $s_{j,v}^{(1)}$ non-zero values over the components $h_{k,u}$ such that $u \in \mathcal{B}_v^{(2)}$.  
In addition, for $k \in \mathcal{B}_j^{(1\rightarrow 2)}$, there are $s_{i,j}^{(1\rightarrow 2)}$ non-zero values over the components $h_{\ell,k}^{(1\rightarrow 2)}$. 
As a result, and since there is not VN elimination,
\begin{equation}
 s_{i,v}^{(2)} = \sum_{j=1}^{S_{m_1}} s_{i,j}^{(1\rightarrow 2)} s_{j,v}^{(2)} ,
\end{equation}
which implies~\eqref{eq:relation_proto}.

\end{proof}
In Proposition~\ref{prop:rel_proto}, assumption 1) is required because various interleaving structures may be used to construct \emph{e.g.} a matrix $H_1$ from a given protograph $\mathcal{S}_1$. This assumption guarantees that the same interleaving structure is used for the CNs of $\mathcal{S}_1$ and the VNs of $\mathcal{S}_{1\rightarrow 2}$.
Further, assumption $2$ guarantees that relation~\eqref{eq:defH2} does not eliminate any VN from the parity check equations in $H_2$. This permits to preserve the code structure that will be characterized by protograph $\mathcal{S}_2$. 
Then, by comparing~\eqref{eq:defH2} and~\eqref{eq:relation_proto}, we observe that there is the same relation between the protographs $\mathcal{S}_{1}$, $\mathcal{S}_{2}$, and between the parity check matrices $H_1$, $H_2$.
Further, according to~\eqref{eq:relation_proto}, the problem of finding a good protograph $\mathcal{S}_2$ for $H_2$ can be reduced to finding the intermediate protograph $\mathcal{S}_{1\rightarrow2}$ that maximizes the threshold of $\mathcal{S}_{2}$.


\subsection{Optimization of the intermediate protograph $\mathcal{S}_{1\rightarrow2}$}\label{sec:optim_int_proto}

The protograph $\mathcal{S}_{1\rightarrow2}$ of size $S_{m_2} \times S_{m_1}$ must be full rank in order to satisfy the rate-adaptive condition defined in Section~\ref{sec:ra_cond}.
However, even if $S_{m_{1}}$ and $S_{m_{2}}$ are small, there is a still a lot of possible protographs $\mathcal{S}_{1\rightarrow2}$. 
This is why, here, we impose that each row of $\mathcal{S}_{1\rightarrow2}$ has either $1$ or $2$ non-zero components, that each column has exactly $1$ non-zero component, and that all the non-zero components are equal to $1$. 
These constraints are equivalent to considering that each row of $\mathcal{S}_2$ is either equal to a row of $\mathcal{S}_1$ or equal to the sum of two rows of $\mathcal{S}_1$.
They limit the number of possible $\mathcal{S}_{1\rightarrow2}$ without being too restrictive. 
They will also make the intermediate matrix $H_{1\rightarrow 2}$ quite sparse, which will help limiting the amount of short cycles in the matrix $H_2$.
Finally, we observe that these constraints provide satisfactory rate-adaptive code constructions in our simulations. The design algorithms described in the remaining of the paper can also be easily generalized to other constraints on the intermediate protograph. 

For the optimization, we then generate all the possible intermediate protographs $\mathcal{S}_{1\rightarrow2}$ that satisfy the above two conditions ($\mathcal{S}_{1\rightarrow2}$ is full rank and each of its rows has either $1$ or $2$ non-zero components), and select the intermediate protograph that maximizes the threshold of the protograph $\mathcal{S}_{2}$ calculated from~\eqref{eq:relation_proto}.

The intermediate protograph $\mathcal{S}_{1\rightarrow2}$ defines the degree distribution of the intermediate matrix $H_{1\rightarrow 2}$.  
It also indicates the rows of $H_1$ that can be combined in order to construct the daughter matrix  $H_2$.
We would like those rows to be combined in the best possible way in order to produce $H_2$. In particular, we would like to avoid both short circles and VN elimination during the construction of $H_2$.
In the following, we propose an algorithm that constructs $H_2$ from these conditions.

\subsection{Algorithm Proto-Circle: connections in $H_{1 \rightarrow 2}$}\label{sec:ra_construction}


\begin{algorithm}[h!]
 \caption{Proto-Circle: construction of the low-rate matrix $H_{2}$}
 \begin{algorithmic}\label{algo:duplication}
 \STATE{\textbf{Inputs}: $H_1$, $\mathcal{S}_1$, $\mathcal{S}_{1\rightarrow 2}$, $K$, $H_2 = \{ \phi \}$ }
 \FOR{$i=1$ to $S_{m_2}$}
 \IF{$i$-th row of $\mathcal{S}_{1\rightarrow 2}$ has two non-zero components $s_{i,j_1}^{(1\rightarrow 2)}$, $s_{i,j_2}^{(1\rightarrow 2)}$ }
 \FOR{$\ell=1$ to $m_1/S_{m_1}$}
 \STATE{Pick $u$ at random in $\mathcal{A}_{j_1}^{(1)}$ and $v_1,\cdots,v_K$ at random in $\mathcal{A}_{j_2}^{(1)}$ such that $\forall k \in \{1,\cdots, K\}$, $\forall w \in \{1,\cdots, m_1\}$, $h_{v_1,w}^{(1)} . h_{v_2,w}^{(1)} = 0$ }
 \STATE{For all $k\in \{1,\cdots,K\}$, count the number $N_{4,k}$ of length-$4$ cycles in $H_2 \cup \{ \underline{h}_{u}^{(1)} + \underline{h}_{v_k}^{(1)} \}$}
 \STATE{For the index $k^{\star}$ that minimizes $N_{4,k}$, do $H_2 \leftarrow H_2 \cup \{ \underline{h}_{u}^{(1)} + \underline{h}_{v_{k^{\star}}}^{(1)}\} $}
 \STATE{Remove $u$ from $\mathcal{A}_{j_1}^{(1)}$ and $v_{k^{\star}}$ from $\mathcal{A}_{j_2}^{(1)}$}
 \ENDFOR
 \ELSE
 \FOR{$\ell=1$ to $m_1/S_{m_1}$}
 \STATE{Pick $u$ at random in $\mathcal{A}_{j_1}^{(1)}$ ($s_{i,j_1}^{(1\rightarrow 2)} \neq 0$) and do $H_2 \leftarrow H_2 \cup \{ \underline{h}_{u}^{(1)} \}$, remove $u$ from $\mathcal{A}_{j_1}^{(1)}$}
 \ENDFOR
 \ENDIF
  \ENDFOR
 \STATE{\textbf{outputs}: $H_2$, $N_4$ (number of length-4 cycles in $H_2$)}
 \end{algorithmic}
\end{algorithm}

In Section~\ref{sec:optim_int_proto}, we selected the intermediate protograph $\mathcal{S}_{1\rightarrow2}$ that gives the protograph $\mathcal{S}_{2}$ with highest threshold.  
We now explain how to construct $H_{1\rightarrow 2}$ in order to follow the degree distribution defined by protograph $\mathcal{S}_{1\rightarrow2}$, but also to limit the amount of short cycles in $H_2$ and to avoid VN elimination.
The algorithm Proto-Circle we propose is described in Algorithm~\ref{algo:duplication}. It constructs one row of $H_2$ at a time by combining rows of $H_1$, which can be regarded as defining the coefficients of the intermediate matrix $H_{1\rightarrow2}$. For each new row of $H_2$, we want to limit the number of short cycles that are added to the parity check matrix $H_2$. 

According to section~\ref{sec:optim_int_proto}, each row of the protograph $\mathcal{S}_{1\rightarrow 2}$ has either $1$ or $2$ non-zero components.
The rows of $\mathcal{S}_{1\rightarrow 2}$ that have $2$ non-zero components indicate that two rows of $H_1$ of some given types should be combined in order to obtain one row of $H_2$. 
More formally, assume that the $i$-th row of $\mathcal{S}_{1\rightarrow 2}$ is such that $s_{i,j_1}^{({1\rightarrow 2})} = 1$ and $s_{i,j_2}^{({1\rightarrow 2})} = 1$ $(j_1\neq j_2)$. 
This means that two rows of $H_1$ of types $A_{j_1}^{(1)}$ and $A_{j_2}^{(1)}$ should be combined in order to obtain one row of $H_2$ of type $A_i^{(2)}$.
For this, we select at random one row $\underline{h}_u^{(1)}$ of $H_1$ of type $A_{j_1}^{(1)}$ and $K$ rows $\underline{h}_{v_1}^{(1)}, \cdots \underline{h}_{v_K}^{(1)}$ of type $A_{j_2}^{(1)}$ such that $\forall k \in \{1,\cdots, K\}$, $\forall w \in \{1,\cdots, m_1\}$, $h_{v_1,w}^{(1)} . h_{v_2,w}^{(1)} = 0$ (binary AND operation).
This condition avoids VN elimination. 
The algorithm counts the number $N_{4,k}$ of length-$4$ cycles that would be added if a new row $\underline{h}_u^{(1)} +  \underline{h}_{v_k}^{(1)}$ was added to $H_2$. 
The number of length-$4$ cycles in $H_2$ is computed with the algorithm proposed in~\cite{mao2001heuristic}. Note that the algorithm can be easily modified to also consider larger cycles.
The algorithm then chooses the row combination that adds least cycles in $H_2$.

Once all the lines of types $A_{j_1}^{(1)}$ and $A_{j_2}^{(1)}$ have been combined, the algorithm passes to the next row of $\mathcal{S}_{1\rightarrow 2}$ with two non-zero components and repeats the same process. It then processes the rows of $\mathcal{S}_{1\rightarrow 2}$ with one non-zero component. 
For instance, assume that row $i'$ of $\mathcal{S}_{1\rightarrow 2}$ has one non-zero component $s_{i',j'}^{(1\rightarrow 2)} = 1$. Then, all the lines of $H_1$ of type $A_{j'}^{(1)}$ are placed into $H_2$. The placement order does not have any influence on the amount of cycles in the matrix $H_2$. 

After constructing all the rows of $H_{2}$, the algorithm counts the total number of length-$4$ cycles in the newly created $H_{2}$. 
At the end, repeating the algorithm Proto-Circle several times allows us to choose the matrix $H_{2}$ with least short cycles.

\color{black}
\subsection{Construction of the set $\mathcal{C}'$}
The intermediate matrix $H_{1\rightarrow2}$ follows the structure of the protograph $\mathcal{S}_{1\rightarrow 2}$. As a result, according to Section~\ref{sec:optim_int_proto}, each of its lines has either $1$ or $2$ non-zero components. 
Further, the algorithm Proto-Circle introduced in Section~\ref{sec:ra_construction} imposes that each row of $H_1$ participates to exactly one combination for the constructions of the rows of $H_1$. 
These two conditions guarantee that $H_{1\rightarrow 2}$ is full-rank so that the rate-adaptive condition presented in Section~\ref{sec:ra_cond} is satisfied.
However, in order to completely define the rate-adaptive code $(H_1,H_{1\rightarrow 2},\mathcal{C}')$, we need to define a set $\mathcal{C}'$ of symbols of $\mathcal{C}$ that will be sent together with the set $\mathcal{U}$ in order to obtain the rate $R_1$.

The set $\mathcal{C}'$ will serve to solve a system of $m_1$ equations $\mathcal{U}$ with $m_1$ unknowns $\mathcal{C} \setminus \mathcal{C}'$. For each syndrome symbol $u_i \in \mathcal{U}$ of degree $d_k$ in $H_{1\rightarrow 2}$, we hence decide to put $d_k-1$ of the $d_k$ CNs connected to $u_i$ into $\mathcal{C}'$. For example, if $u_1 = c_1 \oplus c_2 \oplus c_3$, $c_1$ and $c_2$ may be placed into $\mathcal{C}'$. 
This strategy guarantees that it is always possible to reconstruct the set $\mathcal{C}$ from $\mathcal{U}$ and $\mathcal{C}'$. In the above example, it indeed suffices to recover $c_3$ as $c_3 = u_1 \oplus c_1 \oplus c_2$.

We now count the number of symbols $c_i$ that are placed into $\mathcal{C}'$ with this strategy.
Since each line of  $H_{1\rightarrow2}$ has either $1$ or $2$ non-zero components, we have $d_k=1$ or $d_k=2$. Denote by $\alpha$ the proportion of values $u_k$ of degree $1$. We have the following relation between $m_1,m_2$ and $\alpha$: $$m_1 = \alpha m_2 + 2(1-\alpha)m_2.$$ This gives that $\alpha = 2-\frac{m_1}{m_2}$. 
Further, according to the code construction proposed in Section~\ref{sec:ra_construction}, each $c_i$ participates to exactly one equation $u_j$.
As a result, in the above strategy, the set $\mathcal{C'}$ is composed by $(1-\alpha)m_2 = m_1-m_2$ different values $c_i$, which is exactly what is required by the rate-adaptive construction.

\section{Generalization to several rates}\label{sec:several_rates}
The above method constructs the matrix $H_2$ of rate $R_2 < R_1$ from the matrix $H_1$. 
In order to obtain lower rates $R_T < R_{T-1} < \cdots < R_2 < R_1$, we need to construct the successive matrices $H_t$, $t \in \{2,\cdots, T\}$. 
As initially proposed in~\cite{mheich18WCNC}, the matrices $H_t$ can be constructed recursively from intermediate matrices  $H_{t-1\rightarrow t}$ such that $H_t = H_{t-1\rightarrow t} H_{t-1}$. 
The intermediate matrices $H_{t-1\rightarrow t}$ are constructed by from the method described in Section~\ref{sec:int_matrix}.

However, with the method of Section~\ref{sec:int_matrix}, the rate values $R_2,\cdots, R_T$ are constrained by the size of the initial protograph $\mathcal{S}_1$. 
For a protograph $\mathcal{S}_1$ of size $S_{m_1} \times S_{n}$, the rate granularity is given by \begin{equation} r_g = \frac{R_1}{S_{m_1}}.\end{equation}  
For instance, if $R_1=1/2$ and $\mathcal{S}_1$ is of size $ 4 \times 8$, only rates $R_2 = 3/8$, $R_3 = 1/4$, $R_4 = 1/8$ can be achieved. 
This is why, in this section, we propose two alternatives methods that allow to decrease the rate granularity $r_g$.

\subsection{Protograph extension}\label{sec:proto_ext}
The first method called ``protograph extension'' consists of lifting the mother protograph $\mathcal{S}_1$ by a factor $Z_e$, in the same way as for producing a parity check matrix from a given protograph (see Section~\ref{subsec:code_construction}). This extension produces a protograph $\mathcal{S}_1'$ of size $Z_eS_{m_1} \times Z_e S_{n}$. 
For instance, the protograph 
\begin{equation}\label{eq:petit_proto} \mathcal{S}_1 = \begin{bmatrix}
                                        1 & 2 & 1 & 3 \\
                                        1 & 0 & 2 & 5
                                       \end{bmatrix}
 \end{equation} can be extended as 
 \begin{equation}\label{eq:grand_proto}
 \mathcal{S}_1' = \left[\begin{array}{cccccccc}
1 & 1 & 1 & 2 & 0 & 1 & 0 & 1\\
0 & 1 & 0 & 1 & 1 & 1 & 1 & 2\\
1 & 0 & 1 & 4 & 0 & 0 & 1 & 1\\
0 & 0 & 1 & 1 & 1 & 0 & 1 & 4
\end{array}\right] .
\end{equation}

The protograph $\mathcal{S}_1$ permits to generate an ensemble $\mathcal{H}_1$ of parity check matrices with asymptotic codeword length. According to~\cite[Theorem 2]{richardson2001design}, all the asymptotic parity check matrices in $\mathcal{H}_1$ have the same decoding performance given by the threshold of $\mathcal{S}_1$. 
The extended protograph $\mathcal{S}_1'$ generates a code ensemble $\mathcal{H}_1' \subseteq \mathcal{H}_1$. As a result, the asymptotic matrices in $\mathcal{H}_1'$ have the same decoding performance as the matrices in $\mathcal{H}_1$, and $\mathcal{S}_1$ and $\mathcal{S}_1'$ have the same theoretical threshold.

The above protograph extension allows to consider more rates, since the rate granularity $r_g'$ of $\mathcal{S}_1'$ is given by $r_g' = r_g/Z_e \leq r_g $. 
However, it is not desirable neither to end up with an extended protograph $\mathcal{S}_1'$ of large size, \emph{e.g.} in the order of magnitude of $m_1$. 
Indeed, in this case, the number of possibilities for intermediate protographs $\mathcal{S}_{t-1\rightarrow t}$ would also become very large.  
In addition, it becomes computationally difficult to compute the theoretical thresholds for large protographs. 
As a result, if the size of $\mathcal{S}_1'$ is large, it will be very difficult to optimize the successive protographs $\mathcal{S}_t$ according to the method described in Section~\ref{sec:optim_int_proto}. 
This is why we now we propose a second method that allows to push further the rate granularity improvement. 

\subsection{Anchor rates}\label{subsec:anchor_rates}
In this second method, consider a protograph $\mathcal{S}_1$ of size $S_{m_1} \times S_{n}$.
As a first step, we do the protograph optimization of Section~\ref{sec:optim_int_proto} for all the possible rates
\begin{equation}\label{eq:succ_rates}
 R_t = R_1 - \frac{(t-1)R_1}{S_{m_1}},
\end{equation}
where $t \in \{1,\cdots, S_{m_1}\}$, and $R_{t-1}-R_t = R_1/S_{m_1}$.
This produces a sequence of protographs $\mathcal{S}_t$, and the rates $R_t$ are called the anchor rates. 
We now want to construct all the possible intermediate rates between any $R_{t-1}$ and $R_{t}$, with a rate granularity $r_g = R_1/m_1$. 

According to Section~\ref{sec:optim_int_proto}, the rows of the intermediate protographs $\mathcal{S}_{t-1 \rightarrow t}$ have either one or two non-zero components.
In addition, in order to obtain all the rates $R_t$ defined in~\eqref{eq:succ_rates}, exactly one row of $\mathcal{S}_{t-1 \rightarrow t}$ has two non-zero components. 
This is why, in order to obtain a rate $R_{t-1} - \frac{1}{m_1}$, we propose to combine two rows of the corresponding type in $H_{t-1}$. The resulting matrix contains the considered row combination, as well as all the non-combined rows of $H_{t-1}$. As in the algorithm Proto-Circle described in Section~\ref{sec:ra_construction}, we choose the row combination that minimizes the amount of short cycles that will be added in the resulting matrix.  
Applying this process recursively allows to obtain all rates $R_{t-1} - kR_1/{m_1}$, with $k \in \{1,\cdots, m_1/S_{m_1}\}$, and $m_1/S_{m_1}=Z_1$, where $Z_1$ is the lifting factor.
This approach also guarantees that at rate $R_t$, the resulting matrix follows the structure of protograph $\mathcal{S}_t$.

The anchor rates method allows to obtain a rate granularity $r_g = R_1/m_1$. In the simulation section, we combine both approaches (protograph extension and anchor rates) in order to obtain an incremental code construction that permits to handle a wide range of statistical relations between the source and the side information.

%% file: Simulation.tex
\fig{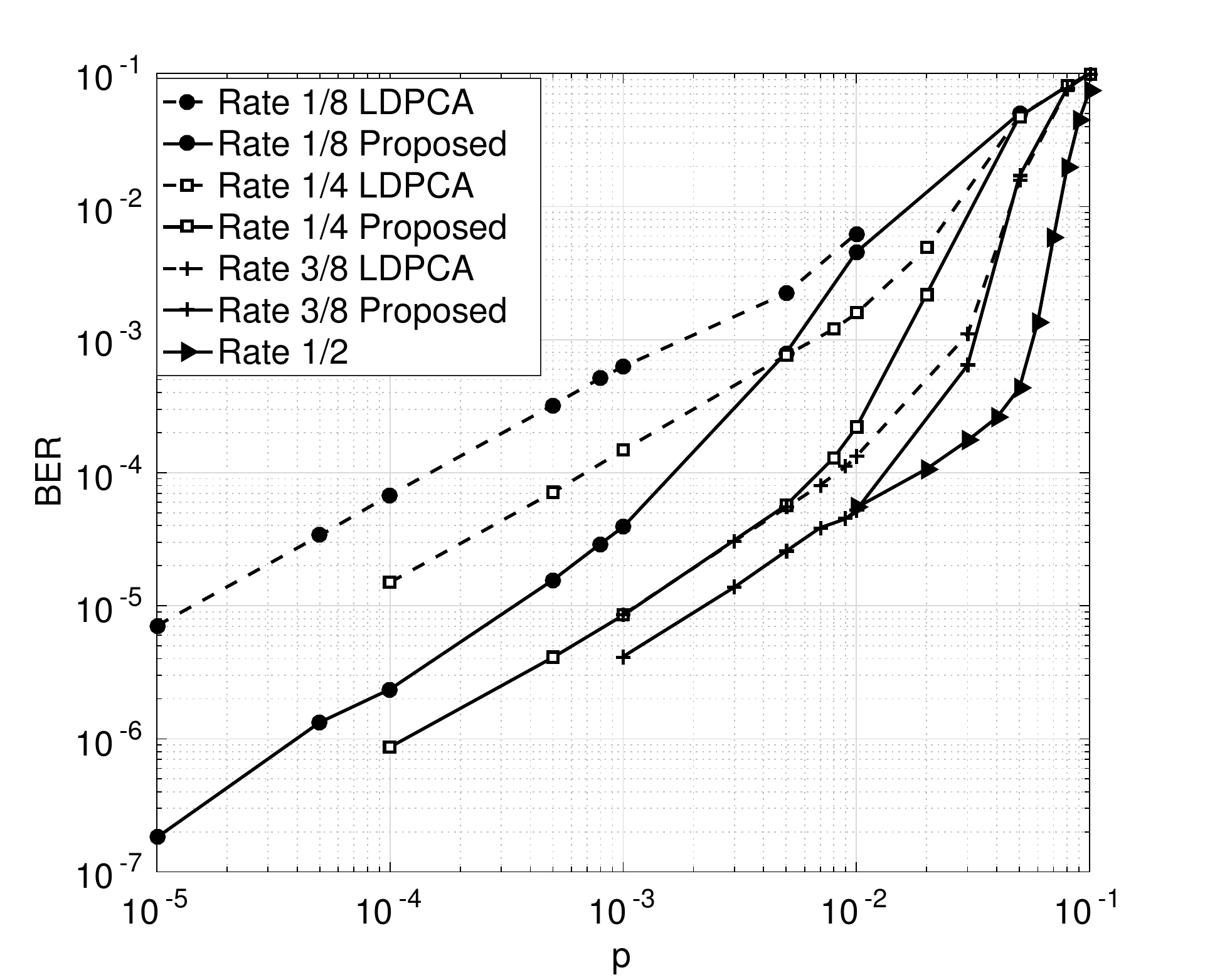}{BER performance of code $\mathcal{C}_1$ with dimension $248 \times 496$ using proposed construction compared with LDPCA}{0.6}{t}
\fig{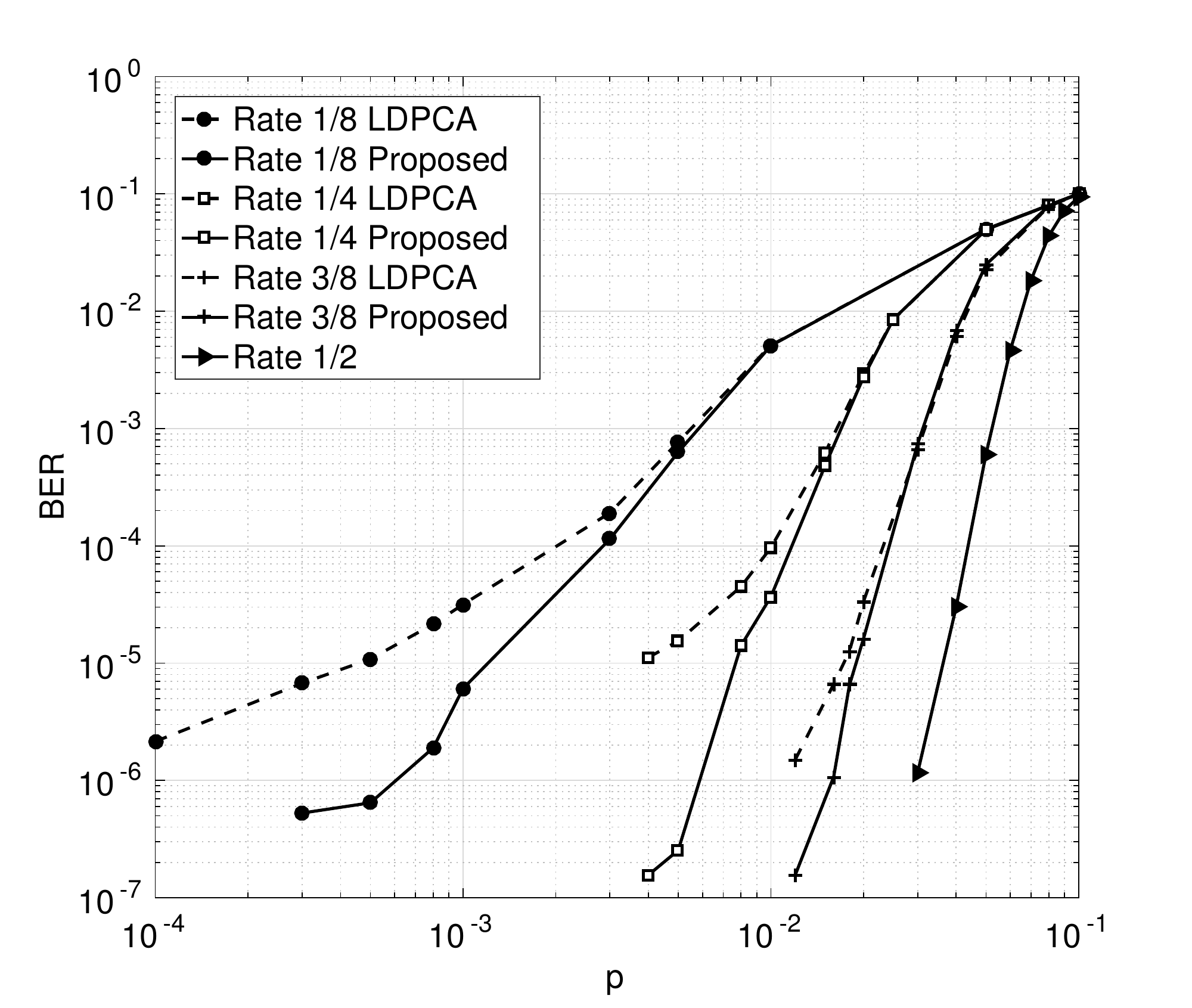}{BER performance of code $\mathcal{C}_2$ with dimension $256 \times 512$ using proposed construction compared with LDPCA}{0.6}{t}
\fig{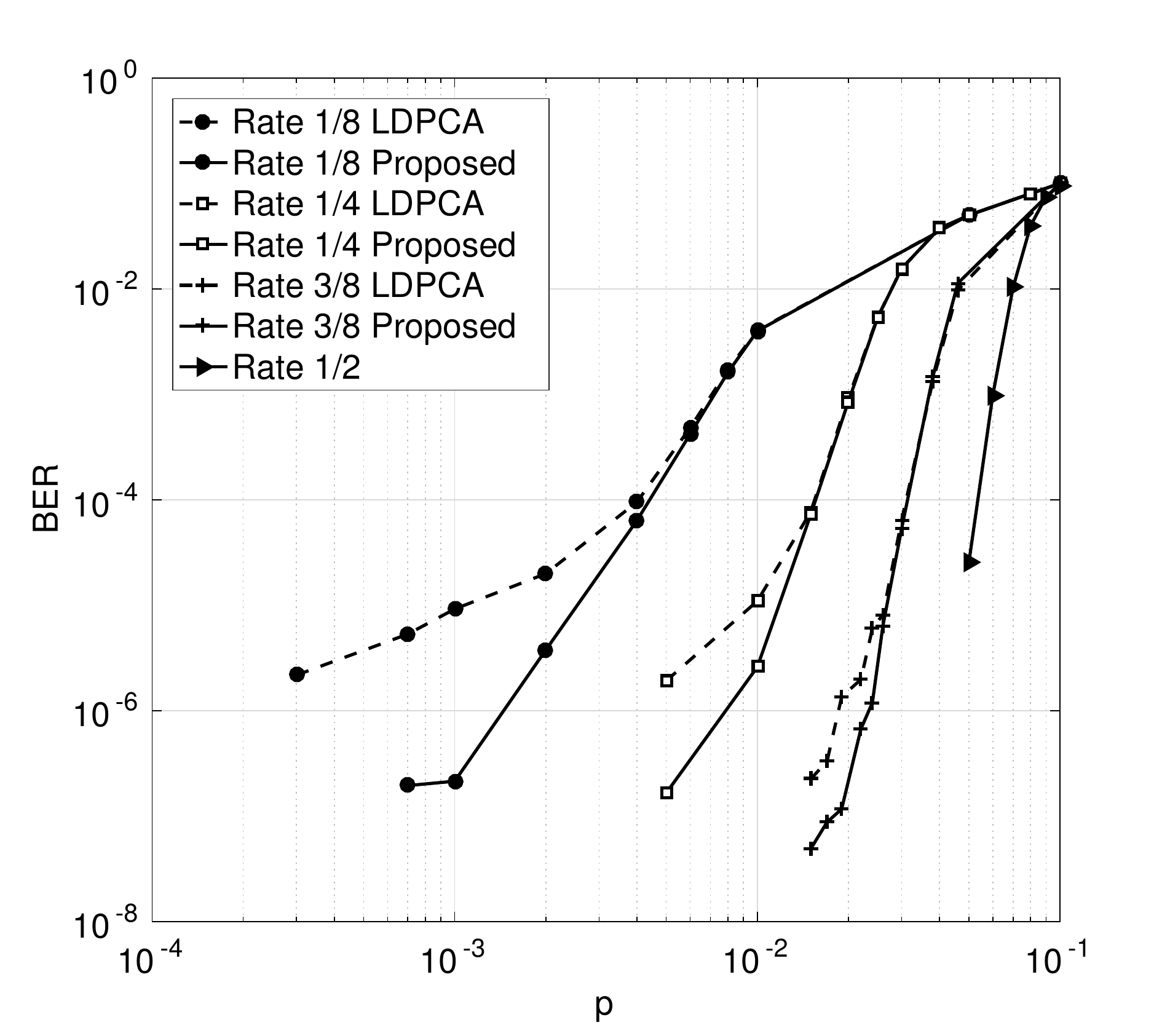}{BER performance of code $\mathcal{C}_3$ with dimension $512 \times 1024$ using proposed construction compared with LDPCA}{0.6}{t}

\begin{table}[h]
\begin{center}
 \begin{tabular}{|l|l|l|}
 \hline
 Rate & LDPCA & Our method \\
 \hline
 $R=3/8$ &  453 & 455\\
 \hline
  $R=1/4$ & 1216 & 737\\
 \hline
  $R=1/8$  & 5361 & 3477\\
  \hline
 \end{tabular}
 \caption{Number of length-4 cycles for code $\mathcal{C}_1$  } \label{fig:N4}
 \end{center}
\end{table}

This section evaluates from Monte Carlo simulations the performance of the proposed rate-adaptive construction. 
We assume a BSC of parameter $p$ and we consider three binary LDPC codes $\mathcal{C}_1$, $\mathcal{C}_2$, $\mathcal{C}_3$ constructed from protographs. These codes are set as mother codes for the initial rate $R = 1/2$. The algorithm introduced in Section~\ref{sec:ra_construction} then produces the corresponding daughter codes for lower rates $3/8$, $1/4$, $1/8$. In the following, we compare the performance of the obtained rate-adaptive codes with LDPCA.

The first code $\mathcal{C}_1$  is of size 248x496. In order to construct $\mathcal{C}_1$, we first obtained the protograph $\mathcal{S}_1$ of size $2\times 4$ in~\eqref{eq:petit_proto}
from the Differential Evolution optimization method described in Section~\ref{subsec:code_construction}.
Differential Evolution was applied by considering $V=60$ elements in the population. This follows~\cite{542711} which suggests to choose $5D < V < 10D$, where in our case, $D=S_nS_m = 8$. In addition, the number of iterations was set as $L=100$, and the maximum degree was set as $d_{\max} = 10$. 
The theoretical threshold of $\mathcal{S}_1$ is equal to $p = 0.094$, which is very close to the maximum value $p=0.11$ that can be considered at rate $1/2$. 
Protograph $\mathcal{S}_1$ was then extended to the protograph $\mathcal{S}_1'$ of size $4\times 8$ in~\eqref{eq:grand_proto} according to the method described in Section~\ref{sec:proto_ext}. 

The parity check matrix of $\mathcal{C}_1$ was constructed from the protograph $\mathcal{S}_{1}'$ by the PEG algorithm~\cite{1377521}. 
We then applied our construction method introduced in Section~\ref{sec:int_matrix} in order to obtain lower rates $3/8$, $1/4$ and $1/8$.
For this, we first needed to decide which rows of the protograph $\mathcal{S}_{opt1}$ should be combined (see Section~\ref{sec:optim_int_proto}) by checking the thresholds of all the possible combinations using Density Evolution.
From Density Evolution, we chose row combinations $A_1^{(1)} + A_3^{(1)}$ for rate $3/8$ and $A_1^{(1)} + A_3^{(1)}, A_2^{(1)} + A_4^{(1)}$ for rate $1/4$,  where the $A_i^{(1)}, i=1,2,\cdots,S_m$, denote the rows of $\mathcal{S}_{1}'$.
From the selected row combinations, we then constructed the corresponding matrices of rate $3/8$, $1/4$, $1/8$ from the algorithm Proto-Circle described in Section~\ref{sec:ra_construction}. This algorithm was applied with $K=20$ and repeated $10$ times in order to choose the low-rate matrices with the least short cycles.

Figure~\ref{fig:proto_simu_n500} shows the Bit Error Rate (BER) performance with respect to the BSC parameter $p$ for the four considered rates for $\mathcal{C}_1$.  
We observe that our code construction performs better than LDPCA at all the considered rates. 
Table~\ref{fig:N4} indeed shows that there are less length-$4$ cycles at rates $1/4$ and $1/8$ in our construction than in the LDPCA matrices.

The second code $\mathcal{C}_2$ is of size $256\times 512$ and it was generated from another protograph
\begin{equation}
\mathcal{S}_{opt2} = \left[\begin{array}{cccccccc}\label{mat:protographGood64}
2 & 1 & 1 & 1 & 0 & 1 & 1 & 0 \\
1 & 2 & 1 & 1 & 1 & 0 & 1 & 1\\
1 & 1 & 2 & 1 & 1 & 1 & 0 & 1\\
1 & 1 & 1 & 2 & 1 & 1 & 1 & 0
\end{array}\right]
\end{equation}
obtained from Differential Evolution and protograph extension.  
The codes of lower rates $3/8$, $1/4$, and $1/8$ were constructed by following the same steps as for  $\mathcal{C}_1$, according to the construction of Section~\ref{sec:int_matrix}.
The BER performance of these codes are shown in Figure~\ref{fig:good64} and compared to LDPCA.
For this case as well, our construction shows better performance than LDPCA. 
Finally, the code $\mathcal{C}_3$ is of size $512\times 1024$ and it was generated from the same protograph $\mathcal{S}_{opt2}$ as $\mathcal{C}_2$. Figure~\ref{fig:good128} shows that for $\mathcal{C}_3$ as well, our algorithm perform better than LDPCA at all the considered rates, with a larger code size.

\fig{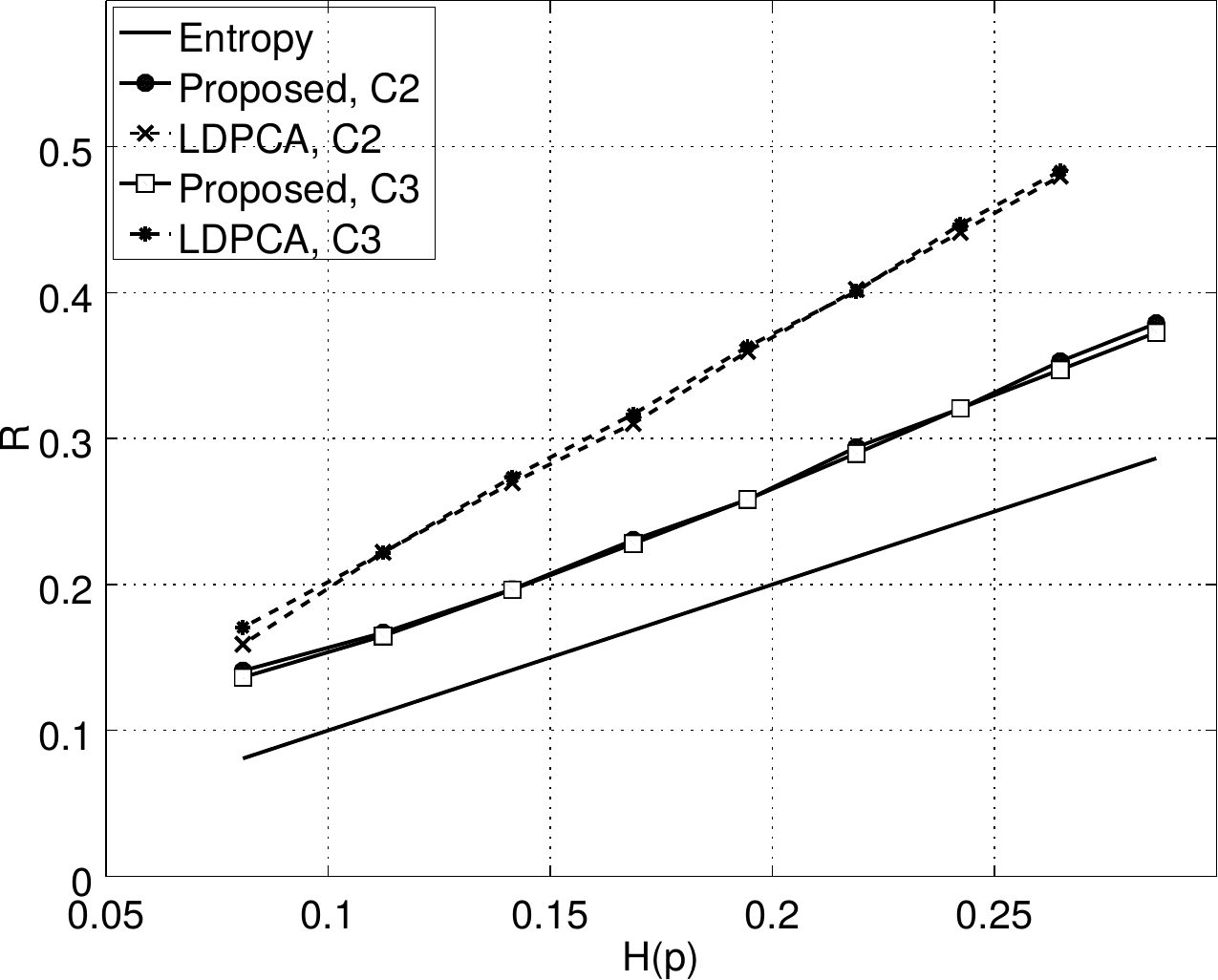}{Required rate $R$ with respect to $H(p)$ for LDPCA and for our method, for codes $\mathcal{C}_2$ and $\mathcal{C}_3$}{0.6}{t}

The curves of Figures~\ref{fig:proto_simu_n500},~\ref{fig:good64},~\ref{fig:good128}, considered the code performance for the anchor rates given in Section~\ref{subsec:anchor_rates}. 
We then applied the method described in Section~\ref{subsec:anchor_rates} to codes $\mathcal{C}_2$ and $\mathcal{C}_3$ in order to obtain rate granularities of $R_1/m_1= 9.8 \times 10^{-4}$ for $\mathcal{C}_2$ and $R_1/m_1= 4.9 \times 10^{-4}$ for $\mathcal{C}_3$, rather than $R_1/{S_{m_1}} = 0.125$.
For this, we considered different values of $p$, and for every considered value, we generated $1000$ couples $(\underline{x}^n,\underline{y}^n)$ from a BSC or parameter $p$.
For every generated couple, we found the minimum rate that permits to decode $\underline{x}^n$ from $\underline{y}^n$ without any error. 
The same kind of analysis was performed in~\cite{varodayan2006rate} and~\cite{cen2009design}, with different criterion to measure the rate needed for a given couple $(\underline{x}^n,\underline{y}^n)$.
In~\cite{varodayan2006rate}, this rate was determined as the minimum rate such that the decoded codeword $\hat{\underline{x}}^n$ verifies $H^T \hat{\underline{x}}^n = \underline{c}^m$, see~\eqref{eq:symdrom_comp}. However, this criterion does not necessarily means that the codeword was correctly decoded ($\hat{\underline{x}}^n$ can be different from $\underline{x}^n$), and this is why we do not consider it here. 
In~\cite{cen2009design}, the required rate was determined as the minimum rate that gives a BER lower than $10^{-6}$. This is equivalent to our criterion, since one uncorrectly decoded bit gives a BER of $2.0 \times 10^{-3}$ for $\mathcal{C}_2$, and of $1.0 \times 10^{-3} $ for $\mathcal{C}_3$.

At the end, Figure~\ref{fig:res_rates} represents the average rates needed for the considered values of $p$ with respect to $H(p)$. 
We first observe that our method shows a loss compared to the optimal rate $H(p)$. This rate loss is expected since we consider relatively short codeword length $512$ for $\mathcal{C}_2$ and $1024$ for $\mathcal{C}_3$.
In addition, for the same codes $\mathcal{C}_2$ and $\mathcal{C}_3$, LDPCA shows a much more important rate loss compared to our method, which was also expected from the results of Figures~\ref{fig:good64} and~\ref{fig:good128}. 
This shows that our construction combined with the anchor rates method is valid and outperforms LDPCA at all the considered values of $p$.

%% file: Conclusion.tex
This paper introduced a novel rate-adaptive construction based on LDPC codes for the Slepian-Wolf source coding problem. 
The introduced construction is based on an optimization of the protographs of the incremental codes constructed at different rates.
It not only allows to optimize the code thresholds at these rates, but also to reduce the amount of short cycles in the obtained parity check matrices. 
The proposed method shows improved performance compared to the standard LDPCA for all the considered codes and at all the considered rates. 
This method may be easily generalized to non-binary LDPC codes. 